\documentclass[twocolumn]{article}

\usepackage[english]{babel}
\usepackage[utf8]{inputenc}
\usepackage[T1]{fontenc}
\usepackage{microtype}
\ifx\directlua\undefined
  \usepackage[T1]{fontenc}
  \usepackage{lmodern}
\else
  \usepackage{fontspec}
\fi

\usepackage[hidelinks]{hyperref}
\usepackage{braket}
\usepackage[inline]{enumitem}
\usepackage[table]{xcolor}
\usepackage{tabularx}

\usepackage{amsthm}

\usepackage[noabbrev, capitalise]{cleveref}
\usepackage{caption}
\usepackage{subcaption}
\usepackage{listings}
\usepackage{tikz}
\usetikzlibrary{graphs}
\usetikzlibrary{positioning}
\usetikzlibrary{arrows.meta}
\usetikzlibrary{fit}

\usepackage{biblatex}
\usepackage{csquotes}

\usepackage[hyphenbreaks]{breakurl}

\makeatletter
\newcommand{\blockid}[1]{\color{black}{\sffamily #1}\hspace{0.5em}}
\lstnewenvironment{basicblock}[2][]{
  \lstset{
    language   = C,
    basicstyle = \ttfamily\small,
    frame      = single,
    framesep   = 0mm,
    framerule  = 0pt,
    #1
  }%
  \def\code@arg{#2}%
  \setbox0\hbox\bgroup%
}
{%
  \egroup\usebox0
  \raisebox{\dimexpr-\dp0+0.25\ht\strutbox\relax}{
    \makebox[\dimexpr-\wd0+\lst@linewidth\relax][r]{
      \blockid{\code@arg}%
    }%
  }%
}
\makeatother

\newtheorem*{proposition}{Proposition}
\newtheorem*{observation}{Observation}

\author{Jørgen Kvalsvik \\ \href{mailto:j@lambda.is}{j@lambda.is}}
\title{Prime Path Coverage in the GNU Compiler Collection}
\date{March 27, 2025}

\begin{document}
  \maketitle
  \begin{abstract}
  We describe the implementation of the prime path coverage support introduced
  the GNU Compiler Collection 15, a structural coverage metric that focuses
  on paths of execution through the program. Prime path coverage
  strikes a good balance between the number of tests and coverage, and
  requires that loops are taken, taken more than once, and skipped.
  We show that prime path coverage subsumes modified
  condition/decision coverage (MC/DC). We improve on the current
  state-of-the-art algorithms for enumerating prime paths by using a
  suffix tree for efficient pruning of duplicated and redundant
  subpaths, reducing it to $O(n^2m)$ from $O(n^2m^2)$, where $n$ is
  the length of the longest path and $m$ is the number of candidate
  paths.  We can efficiently track candidate paths using a few bitwise
  operations based on a compact representation of the indices of
  the ordered prime paths.  By analyzing the control flow graph, GCC
  can observe and instrument paths in a language-agnostic manner, and
  accurately report what code must be run in what order to achieve
  coverage.
\end{abstract}

\section{Introduction}
\label{sec:introduction}
A major limitation in functional testing and dynamic software analysis
is the \emph{path coverage problem}~\cite{ammann2016, larson03,
engler03}, i.e.\ problems can only be detected in executed paths.
Fuzzing~\cite{zalewski-afl} has proven to be an effective technique
for exploring paths and detecting bugs, and there are
algorithms~\cite{li12} that try to generate minimal inputs
for coverage.  There have been proposed hardware extensions
for dynamically expanding into untested paths~\cite{pathexpander07},
and automatic high coverage test generation~\cite{cadar2008}.
Structural coverage analysis on its own remains a useful tool as it
provides increased visibility into the code exercised when
testing~\cite{hayhurst2001}, an objective
exit-criterion~\cite{chilenski2001} for manual testing and test
writing, and is powerful for checking assumptions.  Ball and
Larus~\cite{ball1996} presented efficient path profiling on SPARC
systems in the 1990s, but the industry has never widely adopted path
coverage.  They note that the cost of path profiling can be comparable
to block- and edge profiling while providing far more information, and
can form a basis for profile-guided optimizations.  In this paper we
describe our implementation of the prime path coverage support in the
GNU Compiler Collection version 15.

While Li, Praphamontripong, and Offutt~\cite{li2009} found the
cost/benefit of prime path coverage worse than edge-pair and
definition-use, it is worth noting that prime path coverage still
found more defects than the other coverage criterion used in the
study.  Durelli et al.~\cite{durelli2018} also found prime path
coverage more effective at finding defects than edge-pair coverage at
a moderate increase in test cases.  Prime path coverage \emph{almost}
subsumes edge-pair coverage~\cite{ammann2016, rechtberger22};
self-edges require test paths with repeated vertices which are not
simple.  Since prime paths almost subsume edge-pairs it is reasonable
that it is sensitive to the same problems, which seems supported by Li
et al.~\cite{li2009}.  Stronger coverage criteria have worse
cost/benefit ratio in terms of number-of-bugs, but are more suited to
detect \emph{deeper} defects.

John Regehr~\cite{regehr872} demonstrates a bug in
\cref{fig:path-sensitive-cond} that branch coverage, and even
edge-pair coverage, will not detect, which he names \emph{path
sensitive conditionals}. These subtle interactions between
conditionals is not uncommon in real-world programs.  Regehr also
notes that due to the \emph{path explosion}
problem~\cite{boonstoppel2008} it is infeasible to achieve 100\% path
coverage. Even full path coverage is not, in itself, sufficient to
find \emph{all} bugs in real programs~\cite{regehr386}.

\begin{figure}
  \begin{lstlisting}[language = C, basicstyle = \ttfamily\small]
int silly (int a) {
  int *p = &x;
  if (a > 10) p = NULL;
  if (a < 20) return *p;
  return 0;
}
  \end{lstlisting}
  \caption{
    Assume \lstinline{x} is a global int.  Testing with 5 and 25
    satisfies branch coverage, but does not trigger the bug (dereferencing
    \lstinline{*p} when \lstinline{p} is \lstinline{NULL}). Prime path
    coverage would require both decisions to be \lstinline{true}. The
    example is from Regehr~\cite{regehr872}.
  }
  \label{fig:path-sensitive-cond}
\end{figure}

Even considering all of the limitations noted above, prime path
coverage remains a powerful tool for software testing; it is a strong
criterion that is simple to describe and understand, while
simultaneously subsuming statement, branch, definition-use, and mostly
edge-pair coverage~\cite{rechtberger22, ammann2016}, which allows
testing and development to focus on a single criterion.  Path
coverage, even limited forms such as simple path- and prime path
coverage, is also able to observe \emph{data dependent} infeasible
paths, i.e.\ paths that cannot be taken due to contradictions or
dependent values, defensive guards, and similar constructs.  This
makes path coverage a useful measurement even when not aiming for
full coverage.  While the main objective of testing should be to
verify that the program complies with the functional requirements,
prime path provides strong evidence that the program meets the
requirements \emph{and} that the requirements are complete.

\section{Background}
\label{sec:background}
The tooling and algorithms in this space tend to work on a finite
state machine (FSM) or graph representation of the
programs~\cite{rechtberger22, ammann2016}.  A control flow graph (CFG)
is a graph representation of computation and control flow for a
program module, e.g.\ a function in C.  In the CFG, the vertices, or
\emph{basic blocks}, represent uninterruptible streams of computation
while the edges characterize the control flow between the basic
blocks.  The CFG is the connected and possibly cyclic directed graph
$G = (V, E, v_0, v_x)$ where $V$ is a non-empty finite set of
vertices, $E \subseteq \set{(u,v) \mid u \in V, v \in V}$, and $v_0$,
$v_e$ are the entry and exit vertices so that for all vertices $v \in
V$ there is a path $(v_0, \ldots, v)$ and $(v, \ldots, v_e)$.  The
entry- and exit vertices do no represent any computation, $v_0$ has no
incoming edges, and $v_e$ has no outgoing edges.

A \emph{simple path} is a sequences of vertices $(v_1, \ldots, v_k)$
where all vertices are distinct.  A \emph{simple cycle} is a sequence
of vertices $(v_1, \ldots, v_k)$ where $v_1 = v_k$ and all other
vertices are distinct. Let $P$ be simple paths, $C$ simple cycles, and
$R = P \cup C$; \emph{prime paths} are the maximal objects of $R$.
A prime path is \emph{covered} if its vertices were visited in
sequence during testing.  \cref{fig:simple-prime-path} shows the
simple- and prime paths for a small function with no cycles.  Counting
simple paths is known to be \#P-complete~\cite{valiant1979}, so
enumerating them is at least as hard, and finding prime paths from the
set of simple paths quickly becomes impractical.  Note that while
there are fewer prime paths than simple paths, even small graphs may
have a large number of prime paths~\cite{kaminski2010}.  Prime path
coverage strikes a good balance between defect sensitivity and the
number of required test paths.  Notably, prime path coverage require
loops to be taken, taken more than once, and
skipped~\cite{ammann2016}.

\begin{figure}
  \centering
  \begin{subfigure}{0.33\columnwidth}
    \centering
    \begin{tikzpicture}[>= latex,
      every node/.style = {draw, minimum size = 16pt, circle}]
          \node (a) at (0.5, 5) {1};
          \node (b) at (0.5, 4) {2};
          \node (c) at (0, 3)   {3};
          \node (d) at (1, 3)   {4};
          \node (e) at (0.5, 2) {5};
          \node (f) at (0, 1)   {6};
          \node (g) at (1, 1)   {7};
          \node (h) at (0.5, 0) {8};

          \path [->] (a) edge (b);
          \path [->] (b) edge (c);
          \path [->] (b) edge (d);
          \path [->] (c) edge (e);
          \path [->] (d) edge (e);
          \path [->] (e) edge (f);
          \path [->] (e) edge (g);
          \path [->] (f) edge (h);
          \path [->] (g) edge (h);
    \end{tikzpicture}
  \end{subfigure}%
  \hfill
  \begin{subfigure}{0.66\columnwidth}
    \begin{description}
      \item [Simple paths]
        \begin{flushleft}
          [1\,2] [1\,2\,3] [1\,2\,3\,5] [1\,2\,3\,5\,6]
          [1\,2\,3\,5\,6\,8] [1\,2\,3\,5\,7] [1\,2\,3\,5\,7\,8]
          [1\,2\,4] [1\,2\,4\,5] [1\,2\,4\,5\,6] [1\,2\,4\,5\,6\,8]
          [1\,2\,4\,5\,7] [1\,2\,4\,5\,7\,8] [2\,3] [2\,3\,5]
          [2\,3\,5\,6] [2\,3\,5\,6\,8] [2\,3\,5\,7] [2\,3\,5\,7\,8]
          [2\,4] [2\,4\,5] [2\,4\,5\,6] [2\,4\,5\,6\,8] [2\,4\,5\,7]
          [2\,4\,5\,7\,8] [3\,5] [3\,5\,6] [3\,5\,6\,8] [3\,5\,7] [3\,5\,7\,8]
          [4\,5] [4\,5\,6] [4\,5\,6\,8] [4\,5\,7] [4\,5\,7\,8] [5\,6]
          [5\,6\,8] [5\,7] [5\,7\,8] [6\,8] [7\,8]
        \end{flushleft}
      \item [Prime paths]
        \begin{flushleft}
          [1\,2\,3\,5\,6\,8]
          [1\,2\,3\,5\,7\,8]
          [1\,2\,4\,5\,6\,8]
          [1\,2\,4\,5\,7\,8]
        \end{flushleft}
    \end{description}
  \end{subfigure}

  \caption{The 41 simple paths and 4 prime paths for the CFG for two sequential
    decisions. Executing the prime paths would cover all simple paths.}
  \label{fig:simple-prime-path}
\end{figure}
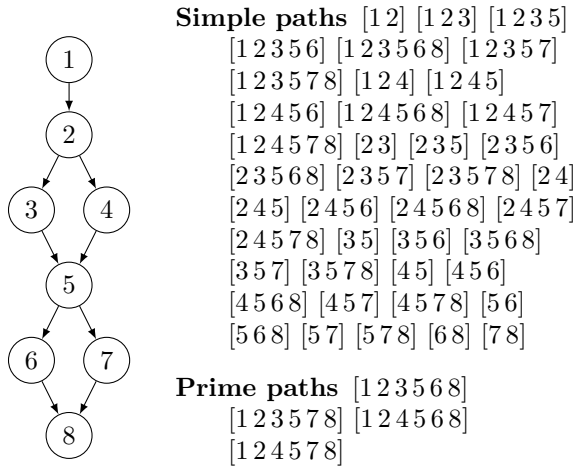

A coverage criterion $C_1$ \emph{subsumes} $C_2$ if satisfying $C_1$
also guarantees satisfying $C_2$~\cite{ammann2016}.  For example,
\emph{vertex coverage}, or node coverage, is the criterion that each
vertex in the CFG shall be visited at least once, and \emph{edge
coverage} is the criterion that each \emph{edge} shall be taken at
least once.  Edge coverage subsumes vertex coverage because taking
every edge in a connected graph guarantees visiting every vertex at
least once.  Prime path coverage subsumes most commonly used coverage
metrics; statement, branch, condition, decision, node, edge, and
definition-use~\cite{rechtberger22, durelli2018}.

\section{Prime path coverage subsumes MC/DC}
\label{sec:ppc-subsumes-mcdc}
The unique property of MC/DC is the independence
criterion~\cite{hayhurst2001}, which states that each condition must
be shown to take on both true and false while independently affecting the
decision's outcome.  A condition can independently affect the outcome
if changing it while keeping the other conditions fixed also changes
the outcome.  Because the independence is shown from the
\emph{combination} of inputs, vertex coverage of a Boolean expression
is not enough to satisfy MC/DC.  A vertex can be visited and still not
contribute towards MC/DC, and it is not obvious that it is subsumed by
prime path coverage.

\begin{proposition}
  Every test vector required for MC/DC is tested with prime path
  coverage.
\end{proposition}

\begin{observation}
  Every simple path is a subpath of a prime path.
\end{observation}

\begin{proof}
  Any Boolean function has a canonical representation as a reduced
  ordered boolean decision diagram (BDD)~\cite{bryant92}, which is a
  rooted acyclic graph.  Each combination of inputs map to a path in
  the BDD, which and MC/DC is achieved by taking a subset of the paths
  in the BDD.  The BDD is acyclic so all paths are simple, and since
  prime path coverage requires taking all maximal simple paths it
  subsumes MC/DC.
\end{proof}

Which subset of paths MC/DC requires depends on the kind of MC/DC.  For
unique-cause MC/DC it is a subset where between pairs of inputs the
final vertex changes when a only a single input condition differs.
For masking MC/DC it is a set of paths that would visit all vertices
while visits are filtered with a masking
function~\cite{kvalsvikmcdc2025}.  The inverse is not true; MC/DC for
a Boolean expression does not imply prime path coverage for the BDD,
as shown in the counter example in \cref{fig:mcdc-subsumed}.

\begin{figure}
  \centering
  \begin{minipage}{0.4\columnwidth}
  \begin{subfigure}{\columnwidth}
    \centering
    \begin{tikzpicture}[>= latex]
      \node [draw, minimum size = 16pt, circle] (a) at (0.5, 4) {$1$};
      \node [draw, minimum size = 16pt, circle] (b) at (1, 3)   {$2$};
      \node [draw, minimum size = 16pt, circle] (c) at (1, 2)   {$3$};
      \node [draw, minimum size = 16pt, circle] (d) at (2, 1)   {$4$};
      \node [draw, minimum size = 16pt, circle, double] (e) at (1, 0)   {$5$};
      \node [draw, minimum size = 16pt, circle, double] (f) at (2, 0)   {$6$};

      \path [->] (a) edge (b);
      \path [->] (a) edge [bend right = 15] (e);
      \path [->] (b) edge (c);
      \path [->] (b) edge [bend left = 15] (d);
      \path [->] (c) edge (e);
      \path [->] (c) edge (d);
      \path [->] (d) edge (e);
      \path [->] (d) edge (f);
    \end{tikzpicture}
    \caption{CFG}
    \label{fig:mcdc-subsumed:cfg}
  \end{subfigure}
  \begin{subfigure}[b]{\columnwidth}
    \centering
    \begin{tabular}{l l}
      $P_1$ & 1 2 3 4 5 \\
      $P_2$ & 1 2 3 4 6 \\
      $P_3$ & 1 2 3 5 \\
      $P_4$ & 1 2 4 5 \\
      $P_5$ & 1 2 4 6 \\
      $P_6$ & 1 5 \\
    \end{tabular}
    \caption{Prime paths}
    \label{fig:mcdc-subsumed:paths}
  \end{subfigure}
  \end{minipage}%
  \begin{subfigure}[t]{0.6\columnwidth}
    \centering
    \rowcolors{1}{}{black!5}
    \begin{tabular}{r cccc c}
         & 1 & 2 & 3 & 4 & Path  \\
      1  & 0 & 0 & 0 & 0 & $P_5$ \\
      2  & 0 & 0 & 0 & 1 & $P_4$ \\
      3  & 0 & 0 & 1 & 0 & $P_5$ \\
      4  & 0 & 0 & 1 & 1 & $P_4$ \\
      5  & 0 & 1 & 0 & 0 & $P_2$ \\
      6  & 0 & 1 & 0 & 1 & $P_1$ \\
      7  & 0 & 1 & 1 & 0 & $P_3$ \\
      8  & 0 & 1 & 1 & 1 & $P_3$ \\
      9  & 1 & 0 & 0 & 0 & $P_6$ \\
      10 & 1 & 0 & 0 & 1 & $P_6$ \\
      11 & 1 & 0 & 1 & 0 & $P_6$ \\
      12 & 1 & 0 & 1 & 1 & $P_6$ \\
      13 & 1 & 1 & 0 & 0 & $P_6$ \\
      14 & 1 & 1 & 0 & 1 & $P_6$ \\
      15 & 1 & 1 & 1 & 0 & $P_6$ \\
      16 & 1 & 1 & 1 & 1 & $P_6$ \\
    \end{tabular}
    \caption{Truth table}
    \label{fig:mcdc-subsumed:inputs}
  \end{subfigure}
  \caption{BDD for the Boolean expression \texttt{a or (b and c) or
    d}.  The double circle vertices 5 and 6 are the true and false
    outcome, respectively.  Rows 1, 2, 5, 7, 9 (paths $P_5$, $P_4$,
    $P_2$, $P_3$, $P_6$) would achieve MC/DC, but row 6 (path $P_1$)
    would have to be included for prime path coverage.
  }
  \label{fig:mcdc-subsumed}
\end{figure}
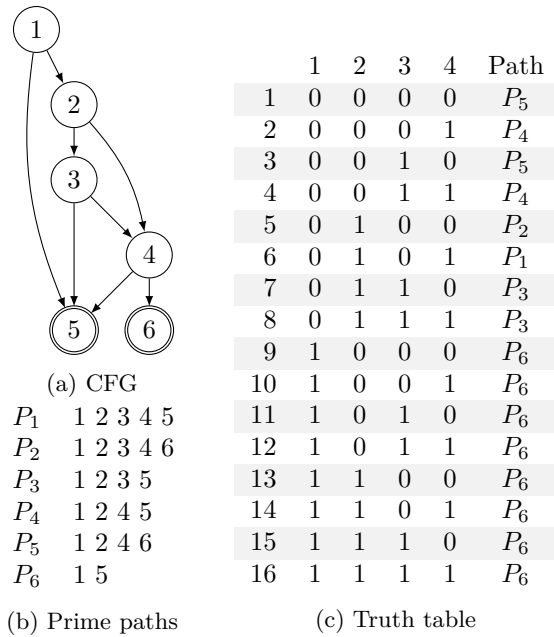

  \section{Enumerating prime paths}
\label{sec:algorithm}
In this section we describe the algorithm in GCC that enumerates the prime
paths.  An efficient algorithm is important as the number of paths
grows very fast with program complexity;  for example, a sequence of
$n$ if-then-else with no nesting, as in
\cref{fig:simple-prime-path}, has $2^n$ prime paths as every
subsequent if-then-else would add another two prime paths to
each of the $2^{n-1}$ prime paths up to it, and the relatively short
function in \cref{fig:binary-search} has 17 prime paths.

We build on two algorithms for prime path enumeration, described by Ammann
and Offutt, and Fazli and Afsharchi.  The algorithm by Ammann and
Offutt~\cite{ammann2016} finds the maximal simple paths and simple
cycles by starting with all single-vertex candidate paths and
progressively extending each path with successors while maintaining
the simple path- and cycle properties, When no more paths can be
extended, the set of paths is pruned by removing all paths that also
appear as subpaths, leaving only the paths that satisfy the
\emph{prime} criterion.  Fazli and Afsharchi's compositional
method~\cite{fazli2019} improve on this for many real-world programs.
The high level steps of their algorithm are to
\begin{enumerate*}[
  label = (\arabic*),
  itemjoin = {{; }},
  itemjoin* = {{, and }}]
  \item compute the component graph of the CFG
  \item generate the prime paths for each component and the component graph
  \item extract different intermediate paths from the components
  \item merge intermediate paths to form the prime paths of the CFG.
  \end{enumerate*}
Fazli and Ebnenasir~\cite{fazli2022} propose a way to use the GPU to
accelerate this design.

\begin{figure}
  \centering
  \begin{subfigure}{\columnwidth}
    \begin{basicblock}{}
int search (int a[], int len, int key) {
    \end{basicblock}
    \begin{basicblock}{1}
  int low = 0;
  int high = len - 1;
    \end{basicblock}
    \begin{basicblock}[backgroundcolor = \color{black!5}]{2}
  while (low <= high) {
    \end{basicblock}
    \begin{basicblock}{3}
    int mid = (low + high) / 2;
    if (a[mid] < key)
    \end{basicblock}
    \begin{basicblock}[backgroundcolor = \color{black!5}]{5}
      low = mid + 1;
    \end{basicblock}
    \begin{basicblock}{6}
    else if (a[mid] > key)
    \end{basicblock}
    \begin{basicblock}[backgroundcolor = \color{black!5}]{7}
      high = mid - 1;
    \end{basicblock}
    \begin{basicblock}{8}
    else
      return mid;
  }
    \end{basicblock}
    \begin{basicblock}[backgroundcolor = \color{black!5}]{4}
  return -1;
    \end{basicblock}
    \begin{basicblock}{9}
}
    \end{basicblock}
    \caption{Code with the basic blocks annotated on the right}
    \label{fig:binary-search:code}
  \end{subfigure}
  \begin{subfigure}{0.5\columnwidth}
    \centering
    \begin{tikzpicture}[>= latex,
      every node/.style = {draw, circle, minimum size = 18pt}]
        \node (a) at ( 0,   6) {1};
        \node (b) at ( 0,   5) {2};
        \node (c) at ( 0,   4) {3};
        \node (d) at ( 1,   4) {4};
        \node (e) at (-0.5, 3) {5};
        \node (f) at ( 0.5, 3) {6};
        \node (g) at (-1,   2) {7};
        \node (i) at ( 0.5, 2) {8};
        \node (h) at ( 1,   1) {9};

        \path [->] (a) edge (b);
        \path [->] (b) edge (c);
        \path [->] (b) edge (d);
        \path [->] (c) edge (e);
        \path [->] (c) edge (f);
        \path [->] (e) edge [bend left = 15] (b);
        \path [->] (f) edge (g);
        \path [->] (f) edge (i);
        \path [->] (g) edge [bend left = 30] (b);
        \path [->] (d) edge (h);
        \path [->] (i) edge (h);
    \end{tikzpicture}
    \caption{CFG}
    \label{fig:binary-search:cfg}
  \end{subfigure}%
  \begin{subfigure}{0.50\columnwidth}
    \centering
    \begin{tabular}{l@{\hspace{1.5\tabcolsep}} l}
      1 2 3 5     & 3 6 7 2 4 9 \\
      1 2 3 6 7   & 5 2 3 5     \\
      1 2 3 6 8 9 & 5 2 3 6 7   \\
      1 2 4 9     & 5 2 3 6 8 9 \\
      2 3 5 2     & 6 7 2 3 5   \\
      2 3 6 7 2   & 6 7 2 3 6   \\
      3 5 2 3     & 7 2 3 6 7   \\
      3 5 2 4 9   & 7 2 3 6 8 9 \\
      3 6 7 2 3   & \\
    \end{tabular}
    \caption{Prime paths}
    \label{fig:binary-search:paths}
  \end{subfigure}
  \label{fig:binary-search}
  \caption{A binary search with its control flow graph and 17 prime
    paths, showing that even short and relatively simple functions can
    have many prime paths.  Note that block 9 is the \emph{action} of
    the return, the transfer of control back to the caller.}
\end{figure}

The effectiveness of the compositional method depends on the
structure of the CFG, as it relies on building on the intermediate
solutions from solving the smaller subgraphs inside the strongly
connected components (SCC).  This is ineffective when the component
graph and the CFG are isomorphic, as in
\cref{fig:simple-prime-path}, or when most of the CFG is within
a single component, as in \cref{fig:fazli-bad}.  The former
happens when there are no loops, and the latter when most of the
function is inside a loop.  In both cases GCC essentially falls
back to the algorithm given by Ammann and Offutt.

\begin{figure}
  \centering
  \begin{tikzpicture}[>= latex]
    \graph[
      no placement,
      nodes = {circle, draw, minimum size = 18pt},
    ] {
      a/1 [x =  0, y = 6];
      b/2 [x =  0, y = 5];
      c/3 [x = -1, y = 4];
      d/4 [x = -1, y = 3];
      e/5 [x = -1.5, y = 2];
      f/6 [x = -0.5, y = 2];
      g/7 [x =  0, y = 1];
      i/8 [x = -2, y = 1];
      h/9 [x =  1, y = 4];

      a -> b -> c -> d -> { e, f };
      e -> i -> [bend left=45] d;
      f -> g -> b;
      b -> h;
    };

    \node (scc) [
      draw,
      dashed,
      inner sep = 0.5em,
      rounded corners,
      fit = (b) (c) (d) (e) (i) (g),
    ] {};
    \node [below right] at (scc.north west) {SCC};
  \end{tikzpicture}
  \caption{Most of the CFG is inside a single component.}
  \label{fig:fazli-bad}
\end{figure}
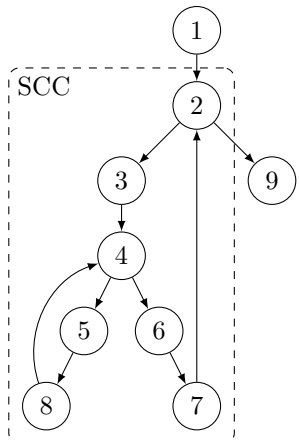

Both algorithms need to filter non-prime paths,
which can be done efficiently with a variant of generalized suffix
trees~\cite{weiner1973, farach1997, ukkonen1995, gusfield1997} over an
integer alphabet of the vertex IDs.  A generalized suffix tree is a
suffix tree for a set of strings and can be constructed as a suffix
tree of the concatenated strings and a special end-of-string
character.  Since prime paths are maximal simple paths or simple cycles
they correspond to the strings in the
suffix tree that do not prefix-match any suffix in the tree.  Most
traditional implementations of suffix trees store offsets into a
string, but the prime paths are generated from the CFG and there is no
explicit string.  It is trivial to construct a string from
candidate paths by concatenation.  Alternatively, the tree can store
substrings directly and not depend on an explicit string. This is
how we implemented the suffix tree in GCC.

\cref{fig:suffix-tree} shows the suffix tree for a simple program and
demonstrates the most important properties: path insertion, subpath
detection, and path reconstruction.  A path $p$ is inserted by
following the path from the root until either $p$ is exhausted or the
current vertex is a leaf, in which case the \emph{final} mark on the leaf
is cleared, the path is extended with the remaining vertices of $p$,
and the new leaf is marked \emph{final}.  The height of the suffix tree is
determined by the \emph{longest} prime path and grows with the length
of the paths, not with the \emph{number} of paths.  A suffix $p' =
\mathbf{tail}~p$ is created and inserted, with the difference that the
final vertex of $p'$ is \emph{not} marked final.  This is repeated
until $p'$ is empty, or until $p'$ fails to create a new vertex, as it
means the remaining suffixes have already been inserted.  Subpath
detection is simple insertion; if a path is extended it is not prime
and the final mark is removed, and if a superpath already existed $p$
will be exhausted before reaching a leaf.  For example, in
\cref{fig:suffix-tree}, inserting the path [1\,4\,2] would not reach
the leaf $4$, and inserting the path [1\,2\,5\,6] would extend
[1\,2\,5] and clear the final mark.  Note that not all leafs are
final; [3\,4\,5] is a suffix of [2\,3\,4\,5] and not prime.  Finally,
there is path reconstruction; prime paths are the paths that end in a
final vertex, and since the tree is \emph{ordered} the paths can be
enumerated in lexicographical order with a depth-first traversal of
the tree.

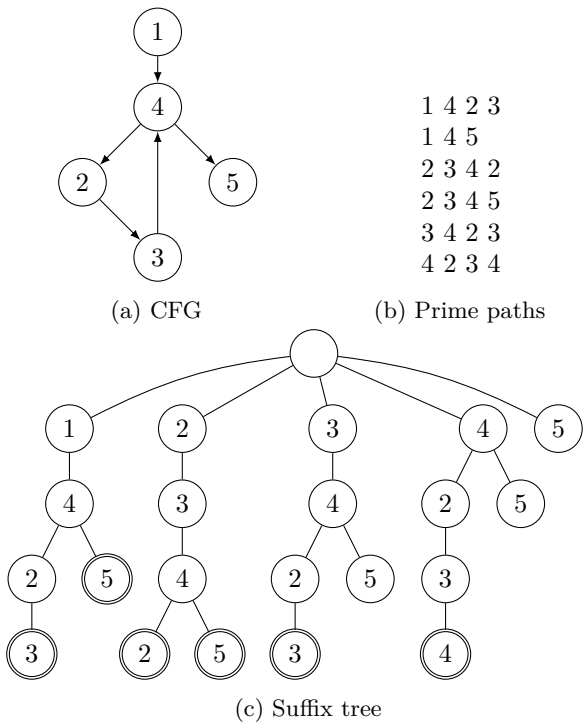
\begin{figure}
  \centering
  \begin{subfigure}[b]{0.5\columnwidth}\centering
    \begin{tikzpicture}[>= latex,
      every node/.style = {circle, draw, minimum size = 18pt}]
        \node (a) at (1, 3) {1};
        \node (d) at (1, 2) {4};
        \node (b) at (0, 1) {2};
        \node (c) at (1, 0) {3};
        \node (e) at (2, 1) {5};

        \path [->] (a) edge (d);
        \path [->] (d) edge (e);
        \path [->] (d) edge (b);
        \path [->] (b) edge (c);
        \path [->] (c) edge (d);
    \end{tikzpicture}
    \caption{CFG}
  \end{subfigure}%
  \begin{subfigure}{0.5\columnwidth}\centering
    \begin{tabular}{l}
      1 4 2 3 \\
      1 4 5   \\
      2 3 4 2 \\
      2 3 4 5 \\
      3 4 2 3 \\
      4 2 3 4 \\
    \end{tabular}
    \caption{Prime paths}
  \end{subfigure}

  \begin{subfigure}{\columnwidth}
    \begin{tikzpicture}[
      every node/.style = {circle, draw, minimum size = 18pt}]

      \node (a) at (3.75, 5) {};

      \node (b1) at (0.5, 4) {1};
      \node (b2) at (2,   4) {2};
      \node (b3) at (4,   4) {3};
      \node (b4) at (6,   4) {4};
      \node (b5) at (7,   4) {5};

      \node (c1) at (0.5, 3) {4};
      \node (c2) at (2,   3) {3};
      \node (c3) at (4,   3) {4};
      \node (c4) at (5.5, 3) {2};
      \node (c5) at (6.5, 3) {5};

      \node (d1) at (0,   2) {2};
      \node [double] (d2) at (1,   2) {5};
      \node (d3) at (2,   2) {4};
      \node (d4) at (3.5, 2) {2};
      \node (d5) at (4.5, 2) {5};
      \node (d6) at (5.5, 2) {3};

      \node [double] (e1) at (0,   1) {3};
      \node [double] (e2) at (1.5, 1) {2};
      \node [double] (e3) at (2.5, 1) {5};
      \node [double] (e4) at (3.5, 1) {3};
      \node [double] (e5) at (5.5, 1) {4};

      \path (a) edge [bend right = 10] (b1);
      \path (a) edge (b2);
      \path (a) edge (b3);
      \path (a) edge (b4);
      \path (a) edge [bend left = 10] (b5);

      \path (b1) edge (c1);
      \path (b2) edge (c2);
      \path (b3) edge (c3);
      \path (b4) edge (c4) edge (c5);

      \path (c1) edge (d1) edge (d2);
      \path (c2) edge (d3);
      \path (c3) edge (d4) edge (d5);
      \path (c4) edge (d6);

      \path (d1) edge (e1);
      \path (d3) edge (e2) edge (e3);
      \path (d4) edge (e4);
      \path (d6) edge (e5);
    \end{tikzpicture}
    \caption{Suffix tree}
  \end{subfigure}
  \caption{A CFG, its prime paths, and the suffix tree after all
    simple paths have been inserted.  Every path from the root to a
    \emph{final} (double) vertex is a prime path, and the ordered
    set of prime paths can be found by an ordered depth-first
    traversal of the tree left to right.  Testing for any subpath is
    done with insertion: if no new vertices are created, the path is a
    subpath.}
  \label{fig:suffix-tree}
\end{figure}

Inserting a path $P$ in the suffix tree is done by inserting the
suffixes $P[1..n], P[2..n], \ldots, P[n]$ where $n = |P|$. Given $m$
prime path candidates we need find the suffix tree, and by extension
set of prime paths and eliminate redundant subpaths, in
$O(n^2m)$. This is an improvement on the pruning step of the algorithm
by Fazli et al~\cite{fazli2019, fazli2022}, which filters using a
nested loop over the set of candidates in $O(n^2m^2)$.  Note that
paths tend to be short and $n$ small; but the \emph{number} of
candidates $m$ grows very fast.

The path explosion problems makes it a practical necessity to limit analysis on
programs that are too complex, i.e.\ programs with too many paths.  The
cost of enumerating the prime paths and emitting instructions grows
with the number of prime paths, as well as the compiled object size.
While there are techniques for
estimating the number of $s-t$ paths~\cite{roberts2007} that could be
used for estimating the prime path count, we employ a simple
pessimistic heuristic; we maintain a running count of the number of paths
and abort whenever it exceeds the given threshold.  The heuristic is
pessimistic as it counts inserts into the suffix tree without applying
corrections for when paths are later subsumed.  This is a pragmatic
and fast solution that slightly over-counts paths and adds very little
overhead, only an increment and a limits check, but might stop analysis on
programs where the total path count is just under the threshold.  The
threshold default is quite high, 250000, and can be set by the user
with a flag.  This is deemed an acceptable solution to the path
explosion problem given the remote likelihood of any tester wanting to
write that many test cases.  As with any complexity problem, the
better solution may be to refactor the program.

  \section{Instrumentation}
\label{sec:instrumentation}
It is necessary to keep track of both completed paths and partially
taken paths.  This can be done efficiently using only a few bitwise
operations, similar to how GCC measures MC/DC~\cite{kvalsvikmcdc2025}.
The prime paths can be ordered lexicographically and a numerical
identifier is assigned to each prime path based on its index in the
ordered set.  For each function, we add the persistent set $P$, which
will be initialized with an empty set the first time the program is run
in an auxiliary file called the counts or .gcda file.  When the
program is run it will, in the function prelude, initialize the
function-local function set $L$.  Each CFG vertex $v$ is extended with
three steps: recording/flushing, discarding, and initializing, in the
order listed below.

\begin{description}
\item[Recording] is updating the persistent counters $P$ in the .gcda file
  with some of the paths in $L$.  There may be more than one
  path which ends in the vertex $v$, and there may be paths that go
  through $v$ which should not be recorded.  $R(v)$ is the set of
  paths that end, and should be recorded, in $v$.

\item[Discarding] is removing the diverged-from paths from the
  candidate set $L$ when taking an edge.  $D(v)$ is the set of paths
  that should be discarded in $v$.

\item[Initializing] is starting the tracking of paths that begin in
  $v$ by adding them to the current candidate set $L$.  $I(v)$ is the set
  of paths that start, and should be initialized, in $v$.
\end{description}

\begin{figure}
  \begin{lstlisting}
P.add(intersection(R(3), L))
L.remove(D(3))
L.add(I(3))
val = getcwd (buffer, size)
if (val != 0)
  \end{lstlisting}
  \caption{Basic block $3$ from \cref{fig:gnu_getcwd} and
    \cref{fig:gnu_getcwd:cfg} with prime path recording and tracking
    as functions on sets.}
  \label{fig:block-with-set-ops}
\end{figure}

These steps are set operations executed immediately upon entering the
vertex $v$. \textbf{Recording} is $P \gets P + (L \cap R(v))$,
\textbf{discarding} is $L \gets L - D(v)$, and
\textbf{initializing} is $L \gets L + I(v)$.  For example, for the
function \lstinline{gnu_getcwd} in \cref{fig:gnu_getcwd}, the vertex
$3$ will be transformed as in \cref{fig:block-with-set-ops}.
Note that for loops the recording and initialization will happen in
the same vertex.  This is fine; since $L$ is initialized to the
empty set, $L \cap R(v)$ will be empty on the first visit of $v$ and
adding it to $P$ becomes a no-op.

The paths and vertices in these examples are taken from
\cref{fig:gnu_getcwd}, and the full table of $R(v)$, $D(v)$ and $I(v)$
is shown in \cref{fig:gnu_getcwd:functions}.  Prime paths are recorded
when entering the last vertex, e.g.\ $R(7) = \set{P_1, P_5}$, so $P_1$
and $P_5$ are recorded when entering vertex $7$.  Prime paths are
discarded when they contain the predecessor $p$ but not the vertex
$v$.  For example, $D(4) = \set{P_1, P_5}$ as there is an edge $(3,4)$
and the vertex $4$ is not in $P_1$ and has no predecessor in $P_5$.
Finally, paths are initialized when entering the first vertex, so
$I(4) = \set{P_5, P_6}$.  These examples demonstrate why the order
listed above is important for these operations. If discarding
happened before recording, $P_5$ would never be covered as it is both
discarded and recorded in $4$. Similarly, if initialization happened
before discarding then $P_5$ would be initialized in $4$ before being
immediately discarded.

\begin{figure}
  \centering
  \begin{basicblock}{}
void *gnu_getcwd () {
  \end{basicblock}
    \begin{basicblock}{1}
  int size = 100;
  void *buffer = alloc (size);
    \end{basicblock}
    \begin{basicblock}[backgroundcolor = \color{black!5}]{8}
  while (1) {
    \end{basicblock}
    \begin{basicblock}{2}
    void *val = getcwd (buffer, size);
    \end{basicblock}
    \begin{basicblock}[backgroundcolor = \color{black!5}]{3}
    if (val != 0)
    \end{basicblock}
    \begin{basicblock}{5}
      return buffer;
    \end{basicblock}
    \begin{basicblock}[backgroundcolor = \color{black!5}]{4}
    size *= 2;
    release (buffer);
    \end{basicblock}
    \begin{basicblock}{6}
    buffer = alloc (size);
    \end{basicblock}
    \begin{basicblock}[backgroundcolor = \color{black!5}]{8}
  }
    \end{basicblock}
    \begin{basicblock}{7}
}
    \end{basicblock}
  \caption{\lstinline{gnu_getcwd} from unix-tree
    2.0.4~\cite{unix-tree} with the basic blocks annotated on the
    right.}
  \label{fig:gnu_getcwd}
\end{figure}

\begin{figure}
  \centering
  \begin{subfigure}[b]{0.5\columnwidth}\centering
    \begin{tikzpicture}[> = latex]
      \graph[
        no placement,
        nodes = {circle, draw, minimum size = 18pt},
      ] {
        a/1 [x =  0, y = 5];
        b/2 [x =  0, y = 4];
        c/3 [x = -1, y = 3];
        e/4 [x = -1, y = 2];
        d/5 [x = -2, y = 2];
        f/6 [x = -1, y = 1];
        h/7 [x = -2, y = 1];
        g/8 [x =  0, y = 0];

        a -> b -> c -> { e, d };
        d -> h;
        e -> f -> g -> b;
      };
    \end{tikzpicture}
    \caption{CFG}
    \label{fig:gnu_getcwd:cfg}
  \end{subfigure}%
  \begin{subfigure}[b]{0.5\columnwidth}\centering
    \begin{tabular}{c l}
       $P_1$ & 1 2 3 5 7 \\
       $P_2$ & 1 2 3 4 6 8 \\
       $P_3$ & 2 3 4 6 8 2 \\
       $P_4$ & 3 4 6 8 2 3 \\
       $P_5$ & 4 6 8 2 3 5 7 \\
       $P_6$ & 4 6 8 2 3 4 \\
       $P_7$ & 6 8 2 3 4 6 \\
       $P_8$ & 8 2 3 4 6 8 \\
    \end{tabular}
    \caption{Prime paths}
    \label{fig:gnu_getcwd:paths}
  \end{subfigure}

  \begin{subfigure}{\columnwidth}
    \rowcolors{1}{}{black!5}
    \begin{tabularx}{\columnwidth}{c *{8}{X}}
      & 1 & 2  & 3  &  4 & 5 & 6  & 7  & 8  \\
$P_1$ & I &    &    &  D &   &    & R  &    \\
$P_2$ & I &    &    &    & D &    &    & R  \\
$P_3$ &   & RI &    &    & D &    &    &    \\
$P_4$ &   &    & RI &    & D &    &    &    \\
$P_5$ &   &    &    & DI &   &    & R  &    \\
$P_6$ &   &    &    & RI & D &    &    &    \\
$P_7$ &   &    &    &    & D & RI &    &    \\
$P_8$ &   &    &    &    & D &    &    & RI \\
    \end{tabularx}
    \caption{}
    \label{fig:gnu_getcwd:functions}
  \end{subfigure}
  \caption{Instrumenting \lstinline{gnu_getcwd}
    (\cref{fig:gnu_getcwd}).  The table in \cref{fig:gnu_getcwd:functions}
    shows which paths are recorded (R), discarded (D), and initialized
    (I) when the vertex is visited.}
  \label{fig:getcwd}
\end{figure}
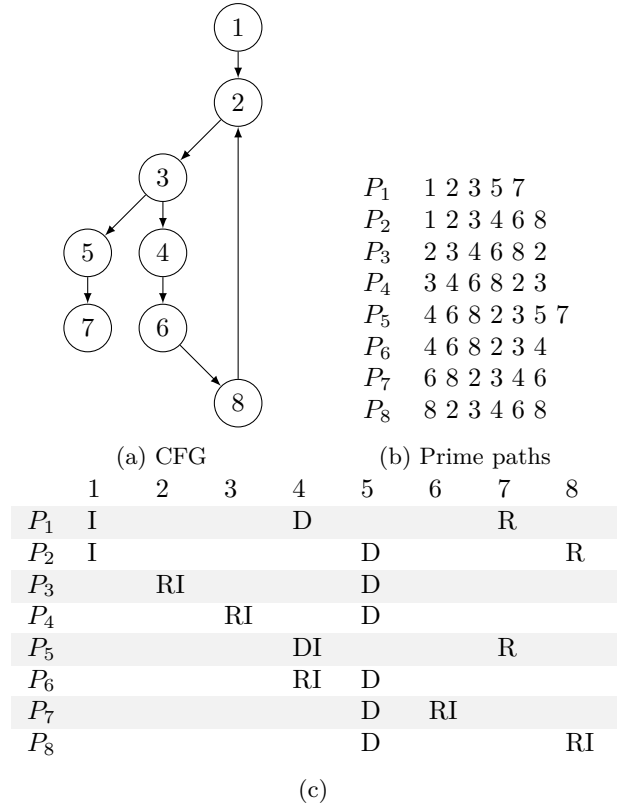

\begin{figure}
  \centering
  \rowcolors{1}{}{black!5}
  \begin{tabularx}{\columnwidth}{c *{3}{>{\centering\arraybackslash}X}}
    $v$ & $B_R(v)$ & $B_D(I)$ & $B_I(v)$ \\
    1 & \ttfamily{00000000} & \ttfamily{00000000} & \ttfamily{00000011} \\
    2 & \ttfamily{00000100} & \ttfamily{00000000} & \ttfamily{00000100} \\
    3 & \ttfamily{00001000} & \ttfamily{00000000} & \ttfamily{00001000} \\
    4 & \ttfamily{00100000} & \ttfamily{00010001} & \ttfamily{00110000} \\
    5 & \ttfamily{00000000} & \ttfamily{11101110} & \ttfamily{00000000} \\
    6 & \ttfamily{01000000} & \ttfamily{00000000} & \ttfamily{01000000} \\
    7 & \ttfamily{00010001} & \ttfamily{00000000} & \ttfamily{00000000} \\
    8 & \ttfamily{10000010} & \ttfamily{00000000} & \ttfamily{10000000} \\
  \end{tabularx}
  \caption{Bitset representation of \cref{fig:gnu_getcwd:functions}}
  \label{fig:bitmasks}
\end{figure}

With the set of prime paths enumerated, the
functions over sets can be efficiently implemented with bitwise
operations as follows; let $B$ be the bitset representation of the set
$P$ where the the $n$th bit $B(n)$ corresponds to the path $P_n$, and
$B_R(v), B_D(v), B_I(v)$ maps to $R(v), D(v), I(v)$ respectively.
We have that $B_R(n) = 1 \textrm{ if } P_n \in R(n)$, $B_D(n) = 1
\textrm{ if } P_n \in D(n)$, and $B_I(n) = 1 \textrm{ if } P_n \in
I(n)$.  All other bits are set to $0$.  A complete table of the bitset
representations of the functions in \cref{fig:gnu_getcwd:functions} is
shown in \cref{fig:bitmasks}.  $B_L$ and $B_P$ are both initialized to
all zeros.  Note that for the bitwise operations it is practical to
invert $B_D$ so applying it as a bitmask \emph{preserves}
non-discarded paths.

We can translate the set functions to work on bitsets, so that for
each vertex $v$;
\begin{enumerate*}[
  label = (\arabic*),
  itemjoin = {{; }},
  itemjoin* = {{, and }}]
  \item \textbf{record} $P \gets P + (L \cap R(v))$ becomes
    \lstinline{P = P | (L & R[v])}
  \item \textbf{discard} $L \gets L - D(v)$ becomes
    \lstinline{L = L & ~D[v]}
  \item \textbf{initialize} $L \gets L + I(v)$ becomes
    \lstinline{L = L | I[v]}.
\end{enumerate*}
Full coverage is achieved when all bits in $P$ are set.

To reduce the runtime and size overhead of instrumentation, certain
unnecessary operations are eliminated.  For example, if no paths are
discarded in $v$ then $B_D(v) = \emptyset{}$, which means $L = L -
B_D(v)$ will have no effect and we do not need to emit instructions for
updating $L$.  The
same techniques apply to the recording and initializing steps.  It can
be seen from the sparseness of the table in
\cref{fig:gnu_getcwd:functions} that many operations can be elided in
practice; instructions have to be emitted if $R(v)$, $D(v)$, or $I(v)$ is
non-empty, but only the instruction corresponding to that step.  For
example, the vertex $4$ will be extended with all the steps, $2$ would
only need \textbf{record} and \textbf{discard}, and $5$ only needs a
single \textbf{discard} instruction to discard multiple prime paths A
complete example of \lstinline{gnu_getcwd} with instrumentation as if
it was written in C is shown in \cref{fig:gnu_getcwd:annotated}.

The number of prime paths of a function is usually much larger than the
native word size or instruction operand sizes.  In GCC 14, released in
2024, the size of the gcov data type used in the .gcda file is 32 or 64
bits, depending on the target architecture.  The bitset of $n$ bits
can be partitioned into $k = \lceil\frac{n}{w}\rceil$ bins of $w$
bits, where $w$ is the number of bits in the gcov type.  It is not
necessary to emit instructions to update bins that are not affected.
For example, for the \lstinline{gnu_getcwd} in \cref{fig:bitmasks}
assuming $w = 4$ would split all bitsets in two.  To perform the
initialization \lstinline{L = L | I[1]} where
\lstinline{I[1] = [0000, 0100]} it is sufficient to only apply the
bitwise-or to the lower half.  This optimization can \emph{greatly}
reduce the size of the compiled object files and improve runtime
performance; in the the tree.c file in unix-tree
2.0.4~\cite{unix-tree}, each vertex typically interacted with up to
10\% of the paths.

\begin{figure}
  \centering
  \begin{basicblock}{}
extern uint P;
void *gnu_getcwd () {
  \end{basicblock}
  \begin{basicblock}{1}
  uint L = 0;
  L |= 00000011;
  int size = 100;
  void *buffer = alloc (size);
  \end{basicblock}
  \begin{basicblock}[backgroundcolor = \color{black!10}]{8}
  while (1) {
  \end{basicblock}
  \begin{basicblock}{2}
    P |= L & 00000100;
    L |= 00000100;
    void *val = getcwd (buffer, size);
  \end{basicblock}
  \begin{basicblock}[backgroundcolor = \color{black!10}]{3}
    P |= L & 00001000;
    L |= 0001000;
    if (val != 0)
  \end{basicblock}
  \begin{basicblock}{5}
      L &= ~11101110;
      goto _return;
  \end{basicblock}
  \begin{basicblock}[backgroundcolor = \color{black!10}]{4}
    P |= L & 00100000;
    L &= ~00010001;
    L |= 00110000;
    size *= 2;
    release (buffer);
  \end{basicblock}
  \begin{basicblock}{6}
    P |= L & 01000000;
    L |= 01000000;
    buffer = alloc (size);
  \end{basicblock}
  \begin{basicblock}[backgroundcolor = \color{black!10}]{8}
    P |= L & 10000010;
    L |= 10000100;
  }
  \end{basicblock}
  \begin{basicblock}{7}
_return:
  P |= L & 00010001;
  return buffer;
}
  \end{basicblock}

  \caption{\lstinline{gnu_getcwd} with instrumentation as-if it was
    written in C.  \lstinline{P} is the persistent bitset, and the
    constant bitmasks are taken from the table in \cref{fig:bitmasks}.
    The label on each block is the is the CFG vertex ID as shown in
    \cref{fig:gnu_getcwd} and \cref{fig:gnu_getcwd:cfg}.}
  \label{fig:gnu_getcwd:annotated}
\end{figure}

  \section{Reporting coverage}
\label{sec:reporting}
The coverage report is printed by the \textbf{gcov} program, which
produces a report from the auxiliary notes and counts files \cite{gcc14-gcov}.  The
notes file, with the extension .gcno, is created by GCC when the
program is compiled, and the counts file, with the extension .gcda, is
created and updated by the instrumented program.  The notes file
stores information about the CFG, function names, line information,
etc., and the counts file store the counters and coverage
measurements.  The prime paths are not stored explicitly in the notes
file as it would be very large, but recomputed from the recorded CFG.
If computing the prime paths for a function was aborted due to
exceeding the path count threshold it is marked as such and gcov will
not attempt to recompute the prime paths for that function.  By using
this information gcov can accurately report on prime path coverage and
describe precisely how to cover the uncovered prime paths.
\cref{fig:report} shows an excerpt of a report on the
\lstinline{gnu_getcwd} function. This is the path $P_2$ = [1 2 3 4 6
8] in \cref{fig:gnu_getcwd:paths}, and gcov prints the lines of source
source, and the basic block associated with it, in the order they must
be executed to cover the prime path. GCC also offers a condensed
line-oriented format intended for machine processing.

\begin{figure}
  \centering
  \begin{lstlisting}[
      language = C,
      basicstyle = \ttfamily\footnotesize,
      linewidth = \columnwidth,
      breaklines = true,
    ]
path 1 not covered:
BB 2:         5:void *getcwd (void *, int)
BB 2:         7:  int size = 100;
BB 2:         8:  void *buffer = alloc (size);
BB 3:        11:    void *value = getcwd (buffer, size);
BB 4:(false) 12:    if (value != 0)
BB 6:        14:    size *= 2;
BB 6:        15:    release (buffer);
BB 7:        16:    buffer = alloc (size);
BB 8:        10:  while (1) {
  \end{lstlisting}
  \caption{Path coverage report for \lstinline{gnu_getcwd} the prime
    path $P_2$ = [1 2 3 4 6 8] in \cref{fig:gnu_getcwd:paths}. The
    columns are the basic block IDs, the edge/transition kind
    (decision), and the source code.  Note that the basic blocks start
    at 2 as GCC reserves 0 and 1 for the entry- and exit blocks.}
  \label{fig:report}
\end{figure}

  \section{Conclusion and future work}
\label{sec:future-work}
This paper describes the implementation by the author of the prime
path coverage support in GCC.  We improve on the algorithms in this
space by utilizing a suffix tree, which ensures fast removal of
subsumed paths as well as providing a compact representation of the
prime paths of a graph.  The instrumented program is reasonably
efficient; it only needs 1 bit per prime path, and the runtime
bookkeeping only needs a few fast bitwise operations.

Finding the prime paths of a graph is an actively researched
topic~\cite{fazli2019, fazli2022}, and improved algorithms would relax the
current limitations and increase path count threshold and support more
complex functions.  Notably, GCC performs worse on graphs with large
SCCs.  Still, the practical limit is not computing the prime paths;
the major slowdown are the later passes where GCC processes the extra
instructions emitted by the instrumentation.  Improvements in later
compiler passes would greatly raise the prime path threshold.

Ball and Larus~\cite{ball1996} show applications of \emph{path
profiling}, which can guide optimization similar to how GCC already
uses edge- and block frequencies in profile-guided optimization.  For
both space and time efficiency the coverage instrumentation uses
bitsets, but can be extended to also record path frequencies.

Our implementation uses a very simple heuristic to determine when
coverage is too expensive, which is inaccurate as an upper bound, and
computationally intensive as it enumerates prime paths up to the
threshold.  Good heuristics and estimation of graph complexity could
be designed to efficiently reject functions.  Fast approximations,
even with reduced accuracy, would improve the responsiveness of the
compiler when faced with large complex programs with many functions
that exceed the threshold.

Prime path coverage is not yet a widely used metric in the industry.
While we have found prime path coverage to be an excellent tool to
evaluate the complexity of code and as driver on where to spend effort
during unit testing, further experiments could measure its
effectiveness at finding defects, cost/benefit ratio, and relationship
with testing the functional requirements.

  \printbibliography
\end{document}